\documentclass[aps,nopacs,nokeys,superscriptaddress,10pt,twoside]{revtex4}

\usepackage{graphicx,epic,eepic,epsfig,amsmath,latexsym,amssymb,verbatim,revsymb,color}

\usepackage{theorem}
\newtheorem{definition}{Definition}

\newtheorem{lemma}[definition]{Lemma}

\newtheorem{theorem}[definition]{Theorem}

\def\squareforqed{\hbox{\rlap{$\sqcap$}$\sqcup$}}
\def\qed{\ifmmode\squareforqed\else{\unskip\nobreak\hfil
\penalty50\hskip1em\null\nobreak\hfil\squareforqed
\parfillskip=0pt\finalhyphendemerits=0\endgraf}\fi}
\def\endenv{\ifmmode\;\else{\unskip\nobreak\hfil
\penalty50\hskip1em\null\nobreak\hfil\;
\parfillskip=0pt\finalhyphendemerits=0\endgraf}\fi}
\newenvironment{proof}{\noindent \textbf{{Proof.~} }}{\qed}

\mathchardef\ordinarycolon\mathcode`\:
\mathcode`\:=\string"8000
\def\vcentcolon{\mathrel{\mathop\ordinarycolon}}
\begingroup \catcode`\:=\active
  \lowercase{\endgroup
  \let :\vcentcolon
  }

\newcommand{\nc}{\newcommand}
\nc{\rnc}{\renewcommand}
\nc{\beq}{\begin{equation}}
\nc{\eeq}{{\end{equation}}}
\nc{\beqa}{\begin{eqnarray}}
\nc{\eeqa}{\end{eqnarray}}
\nc{\lbar}[1]{\overline{#1}}
\nc{\bra}[1]{\langle#1|}
\nc{\ket}[1]{|#1\rangle}
\nc{\ketbra}[2]{|#1\rangle\!\langle#2|}
\nc{\braket}[2]{\langle#1|#2\rangle}
\nc{\proj}[1]{| #1\rangle\!\langle #1 |}
\nc{\avg}[1]{\langle#1\rangle}
\rnc{\max}{\operatorname{max}}
\nc{\Rank}{\operatorname{Rank}}
\nc{\smfrac}[2]{\mbox{$\frac{#1}{#2}$}}
\nc{\tr}{\operatorname{Tr}}
\nc{\ox}{\otimes}
\nc{\dg}{\dagger}
\nc{\dn}{\downarrow}
\nc{\cA}{{\cal A}}
\nc{\cB}{{\cal B}}
\nc{\cC}{{\cal C}}
\nc{\cD}{{\cal D}}
\nc{\cE}{{\cal E}}
\nc{\cF}{{\cal F}}
\nc{\cG}{{\cal G}}
\nc{\cH}{{\cal H}}
\nc{\cI}{{\cal I}}
\nc{\cJ}{{\cal J}}
\nc{\cK}{{\cal K}}
\nc{\cL}{{\cal L}}
\nc{\cM}{{\cal M}}
\nc{\cN}{{\cal N}}
\nc{\cO}{{\cal O}}
\nc{\cP}{{\cal P}}
\nc{\cR}{{\cal R}}
\nc{\cS}{{\cal S}}
\nc{\cT}{{\cal T}}
\nc{\cX}{{\cal X}}
\nc{\cZ}{{\cal Z}}
\nc{\csupp}{{\operatorname{csupp}}}
\nc{\qsupp}{{\operatorname{qsupp}}}
\nc{\var}{\operatorname{var}}
\nc{\rar}{\rightarrow}
\nc{\lrar}{\longrightarrow}
\nc{\polylog}{\operatorname{polylog}}
\nc{\1}{{\openone}}
\nc{\sign}{{\operatorname{sign}}}

\def\a{\alpha}
\def\b{\beta}
\def\g{\gamma}
\def\d{\delta}

\def\ve{\varepsilon}

\nc{\RR}{{{\mathbb R}}}
\nc{\CC}{{{\mathbb C}}}
\nc{\FF}{{{\mathbb F}}}
\nc{\NN}{{{\mathbb N}}}
\nc{\ZZ}{{{\mathbb Z}}}
\nc{\PP}{{{\mathbb P}}}
\nc{\QQ}{{{\mathbb Q}}}
\nc{\UU}{{{\mathbb U}}}
\nc{\EE}{{{\mathbb E}}}
\nc{\id}{{\operatorname{id}}}

\nc{\be}{\begin{equation}}
\nc{\ee}{\end{equation}}
\nc{\bea}{\begin{eqnarray}}
\nc{\eea}{\end{eqnarray}}
\nc{\<}{\langle}
\rnc{\>}{\rangle}
\nc{\Hom}[2]{\mbox{Hom}(\CC^{#1},\CC^{#2})}
\nc{\rU}{\mbox{U}}

\nc{\ob}[1]{#1}

\nc{\eq}[1]{(\ref{eq:#1})}
\def\ot{\otimes}

\begin{document}

\title{Counterexamples to additivity of minimum output p-R\'{e}nyi
       entropy for p close to 0}

\author{Toby Cubitt}
\affiliation{Department of Mathematics, University of Bristol, Bristol BS8 1TW, U.K.}

\author{Aram W. Harrow}
\affiliation{Department of Computer Science, University of Bristol, Bristol BS8 1UB, U.K.}

\author{Debbie Leung}
\affiliation{Institute for Quantum Computing, University of Waterloo,
             Waterloo N2L 3G1, Ontario, Canada}

\author{Ashley Montanaro}
\affiliation{Department of Computer Science, University of Bristol, Bristol BS8 1UB, U.K.}

\author{Andreas Winter}
\affiliation{Department of Mathematics, University of Bristol, Bristol BS8 1TW, U.K.}
\affiliation{Centre for Quantum Technologies, National University of Singapore,
 2 Science Drive 3, Singapore 117542}
\email{a.j.winter@bris.ac.uk}

\date{22nd January 2008}

\begin{abstract}
  Complementing recent progress on the additivity conjecture of
  quantum information theory, showing that the minimum output
  $p$-R\'{e}nyi entropies of channels are not generally additive for
  $p>1$, we demonstrate here by a careful random selection argument
  that also at $p=0$, and consequently for
  sufficiently small $p$, there exist counterexamples.

  An explicit construction of two channels from $4$ to $3$ dimensions
  is given, which have non-multiplicative minimum output rank;
  for this pair of channels, numerics strongly suggest that
  the $p$-R\'{e}nyi entropy is non-additive for all $p \lesssim 0.11$.
  We conjecture however that violations of additivity exist for all $p<1$.
\end{abstract}

\maketitle

\section{Introduction and Definitions}
\label{sec:intro}
For a quantum channel (i.e.~a completely positive and trace preserving
linear map) ${\cal N}$ between finite quantum systems, and $p \geq 0$,
define
\[
  S_p^{\min}({\cal N}) := \min_{\rho} \frac{1}{1-p}\log\tr({\cal N}(\rho))^p,
\]
where the minimisation is over all states (normalised density operators)
on the input space of ${\cal N}$.
The quantity $S_p(\sigma) = \frac{1}{1-p}\log\tr\sigma^p$ is known as
$p$-R\'{e}nyi entropy of the state $\sigma$ ($0<p<\infty$ and $p\neq 1$),
with the definition extended to $p=0,1,\infty$ by taking limits;
$S_1(\sigma) = S(\sigma) = -\tr\sigma\log\sigma$ is the von Neumann entropy.
$S_\infty(\sigma) = -\log \|\sigma\|_\infty$ is the min-entropy, and
$S_0(\sigma) = \log{\rm rank}\,\sigma$.
Due to the concavity of the R\'{e}nyi entropies in $\rho$, the
minimum in the above definition is attained at a pure input state
$\rho = \proj{\psi}$.

The additivity problem is the question whether for all channels ${\cal N}_1$
and ${\cal N}_2$, it holds that
\begin{equation}
  \label{eq:add}
  S_p^{\min}({\cal N}_1\ox{\cal N}_2) \stackrel{?}{=} S_p^{\min}({\cal N}_1) + S_p^{\min}({\cal N}_2).
\end{equation}
Note that the direction ``$\leq$'' here is trivial, so proofs and
counterexamples have to concentrate on the direction ``$\geq$''.
This was indeed proved for special channels and some $p$; for example,
it is known for $p \geq 1$ if one of the channels is
entanglement-breaking~\cite{shor,king1}, unital on a qubit
space~\cite{king2}, or depolarising of any dimension~\cite{king3};
in addition for a number of other cases.
King~\cite{king4} has furthermore shown that it holds for $p<1$
if one of the channels is entanglement-breaking.
Holevo and Werner~\cite{HW}
exhibited the first counterexamples to eq.~(\ref{eq:add}), for $p>4.79$.
It was demonstrated recently~\cite{winter-add,hayden-add} that
for every $p>1$ there exist channels violating eq.~(\ref{eq:add}).

Here we show that eq.~(\ref{eq:add}) is also false at $p=0$, and by
continuity of 
$S_p$ in $p$, it is thus violated for all $p \leq p_0$ with some small
but positive $p_0$. Since $S_0(\sigma)$ is the logarithm of the rank of the
density matrix $\sigma$,
so $S_0^{\min}({\cal N})$ is the logarithm of the minimum output rank of the channel,
i.e.~of the smallest rank of an output state. In the next Section we prove our
main existence result of counterexamples, in Section~\ref{sec:extension}
we exhibit an explicit example, and in Section~\ref{sec:larger-p} we explore
up to which $p<1$ we can violate additivity of $S_p^{\min}$.

\section{Main result}
\label{sec:main}

\begin{theorem}
\label{thm:main}
If $d_A>2$, $d_B > 2$ and $d_Ad_B$ is even then there exist
quantum channels ${\cal N}_1$, 
${\cal N}_2$ with $d_A$-dimensional input spaces and $d_B$-dimensional
output spaces, such that
\[
  S_0^{\min}({\cal N}_1) = S_0^{\min}({\cal N}_2) = \log d_B,
\]
but 
\[
  S_0^{\min}({\cal N}_1\ox{\cal N}_2) \leq \log(d_B^2-1) < 2\log d_B.
\]
\end{theorem}

\begin{proof}
Our approach is the following: let $\rho_{AB} = (\id\ox{\cal N})\Psi$
be a Choi-Jamio\l{}kowski state of the channel ${\cal N}$, with a particular choice
of reference state $\ket{\Psi} \in A \ox A' = \CC^{d_A} \ox
\CC^{d_A}$. Note that while usually people 
use a fixed maximally entangled state, for the isomorphism it is
sufficient that it is of maximal Schmidt rank.
In \cite{hayden-add,winter-add}, additivity counterexamples were found
for $p>1$ by choosing $\cN$ randomly subject to a certain constraint.  Our
approach will be instead to choose the Choi-Jamio\l{}kowski state randomly,
again subject to a certain constraint that helps guarantee the
additivity counterexample.

First note that ${\cal N}(\varphi)$ has maximal rank $d_B$ for every
input state $\varphi$ 
iff the orthogonal complement of $\rho$ doesn't contain any product
vectors, i.e.~for all pure states $\ket{\varphi}\in A$, $\ket{\psi}\in B$,
\be
  \label{eq:main-insight}
  \tr\bigl( \rho_{AB} (\varphi\ox\psi) \bigr) \neq 0.
\ee
The easy justification of this is as follows: in Appendix A we show that
the action of the channel ${\cal N}$ can be written
\be
  \label{eq:c-j}
  {\cal N}(\varphi) = \tr_A \left[ \rho_{AB}
              \left( \rho_A^{-1/2}U^\dagger\overline{\varphi}U\rho_A^{-1/2} \ox \1 \right) \right],
\ee
where $\overline{\cdot}$ denotes the complex conjugate with respect to a 
fixed computational basis and $U$ is a unitary depending on $\Psi$
(see Appendix A for details). Full rank of the output means that for
all pure states $\varphi$, $\psi$,
\[\begin{split}
  0 &\neq \tr\bigl( {\cal N}(\varphi) \psi \bigr) \\
    &=    \tr \left[ \rho_{AB}
            \left( \rho_A^{-1/2}U^\dagger\overline{\varphi}U\rho_A^{-1/2} \ox \psi \right) \right] \\
    &\propto  \tr\bigl[ \rho_{AB} (\varphi'\ox\psi) \bigr],
\end{split}\]
where we used eq. \eq{c-j} and the fact that
$\rho_A^{-1/2}U^\dagger\ket{\overline{\varphi}}$ is, up to normalisation,
another pure state $\ket{\varphi'}$.
Note that \emph{any} unitary $U$ on $A$ will serve to create a channel,
so we shall fix it to be the identity from now on -- this is only a matter
of redefining $\Psi_{AA'}$, which we can do if only given $\rho_{AB}$.

So, our task is to find two states $\rho_{AB}$ and $\sigma_{A'B'}$
on $A \ox B$
with this property, such that $\omega_{AA'\,BB'} = \rho_{AB}\ox\sigma_{A'B'}$
\emph{does have} a product state in its orthogonal complement; we'll choose it
to be the maximally entangled state $\Phi_{AA'} \ox \Phi_{BB'}$.
Then the condition we seek to enforce is
\[
  0 = \tr\bigl( (\rho_{AB} \ox \sigma_{A'B'}) (\Phi_{AA'} \ox \Phi_{BB'}) \bigr)
    = \frac{1}{d_Ad_B} \tr (\rho\,\sigma^\top),
\]
where $\top$ signifies the matrix transpose.  Note that the channel
input will not be $\Phi_{AA'}$, but rather the normalised version of
$\left(\sqrt{\rho_A} \ot \sqrt{\sigma_{A'}}\right)\ket{\Phi}_{AA'}$.

What we will do is simply pick $\rho$ to be the (normalised)
projection onto a $d_Ad_B/2$-dimensional random subspace, drawn
according to the unitary invariant measure on $AB$,
and $\sigma^\top$ the (normalised) projection onto the orthogonal
complement of $\rho$:
\[
  \rho   = \frac{2}{d_Ad_B}\Pi,\quad
  \sigma = \frac{2}{d_Ad_B}(\1-\Pi^\top).
\]
This enforces the condition $\tr(\rho\,\sigma^\top) = 0$ deterministically,
while both the supporting subspaces of $\rho$ and $\sigma$ are
individually uniformly random. Thus we are done once we prove Lemma~\ref{lemma:Pr-0},
stated below, since it implies for large enough $d_A$ and $d_B$, that
with probability $1$ neither the orthogonal complement of $\rho$ nor
that of $\sigma$ (which are themselves uniformly random subspaces)
contains a product vector.
\end{proof}

\begin{lemma}
\label{lemma:Pr-0}
Let $\Pi$ be a uniformly random projector in $\CC^{d_A} \ox \CC^{d_B}$
of rank $d_E$ such that $d_Ad_B > d_A + d_B + d_E - 2$. Then,
\be
  \Pr_\Pi\left\{ \exists \varphi_A\in\CC^{d_A},\varphi_B\in\CC^{d_B}
 \tr\bigl( (\varphi_A\ox\varphi_B)\Pi
    \bigr) = 1 \right\} =  0.
\label{eq:conc-bound}
\ee
In words: the probability that a random subspace of ``small''
dimension contains a product state, is zero.
\end{lemma}

Note that if $d_Ad_B$ is even
then $d_E=d_Ad_B/2$ is an integer, and so the inequality
$(d_A-2)(d_B-2)>0$ can be rearranged to obtain
$d_Ad_B > d_A + d_B + d_E - 2$, thus justifying the application to
Theorem~\ref{thm:main}.

\begin{proof}
Geometrically, we want to show that the probability for a random
subspace of dimension $d_E$ to contain a product state, is zero.
Using the isomorphism between bipartite vectors and $d_A\times d_B$-matrices
(which identifies Schmidt rank with matrix rank)~\cite{sub-schmidt},
we can reformulate the task as describing the $d_E$-dimensional
subspaces of $d_A \times d_B$-matrices not containing any
nonzero elements of rank $1$. In other words, we are interested in the
set of subspaces
intersecting the \emph{determinantal variety} of vanishing
$2\times 2$-minors only in the zero matrix. The
dimension of this variety -- known as the Segre embedding -- is
easily seen to be $d_A + d_B - 1$, so a generic subspace of
dimension $d_E \leq d_A d_B - (d_A + d_B - 1) = (d_A-1)(d_B-1)$
will not intersect it except trivially, by standard algebraic-geometric
arguments~\cite{Eisenbud,Landsberg}; a more explicit argument for
this fact was given recently by Walgate and Scott~\cite{WalgateScott}.
\end{proof}

\section{An explicit construction in small dimension}
\label{sec:extension}
Since our additivity violation takes the form of
only a single zero eigenvalue in the two-copy output, it is strongest
when the channel dimensions are smallest.  Indeed, violations for
large dimension can be constructed from channels from small dimension
by tensoring the channel with a trivial channel, such as a completely
depolarising channel.  Thus, we are most interested in finding
counterexamples with small dimension.

One such counterexample, with $d_A=4$ and $d_B=3$ is described here.
Based on the constructions in~\cite{sub-schmidt}, and indeed a slight variation
of it, we show now -- using the same methodology as above -- how to
construct two channels ${\cal N}_i:{\cal B}(\CC^4) \rightarrow {\cal B}(\CC^3)$
($i=1,2$) such that
\[
  S_0^{\min}({\cal N}_1) = S_0^{\min}({\cal N}_2) = \log 3,
  \ \text{ but }\ 
  S_0^{\min}({\cal N}_1 \ox {\cal N}_2) \leq \log 8 < 2\log 3.
\]
These happen to be the smallest dimensions that satisfy Lemma~\ref{lemma:Pr-0}.

As we have discussed above, we describe them via their Choi-Jamio\l{}kowski
states $\rho_{AB}$ and $\sigma_{AB}$ (with $A$ and $B$ being a $4$-
and $3$-dimensional system, respectively) such that $\tr\rho\sigma^\top = 0$
and neither $\rho$ nor $\sigma$ contains a product state in the
respective orthogonal complement of their supports.

Resorting to the supporting subspaces of $\rho$ and $\sigma^\top$,
denoted $R,S < A\ox B$, respectively,
we have nothing to do but choose them to be orthogonal and 
of dimension $6$, such that neither contains a product state.

Using the customary notation of vectors in $\CC^4 \times \CC^3$
as $3 \times 4$ matrices~\cite{sub-schmidt},
and with $\omega = e^{2\pi i/3}$, we let $R$ be spanned by
\[
  \left[\begin{array}{rrrr}
            &   &    &  \\
          1 &   &    &  \\
            & 1 & 0  & 0 
        \end{array}\right],\ 
  \left[\begin{array}{rrrr}
          1 &   & &  \\
            & 1 & &  \\
            &   & 1 & 0 
        \end{array}\right],\ 
  \left[\begin{array}{rrrr}
          1 &        & &  \\
            & \omega & &  \\
            &        & \omega^2 & 0 
        \end{array}\right],\ 
  \left[\begin{array}{rrrr}
          0 & 1 & &  \\
            &   & \omega^2 &  \\
            &   &          & \omega
        \end{array}\right],\ 
  \left[\begin{array}{rrrr}
          0 & 0 & 1  &    \\
            &   &    & -1 \\
            &   &    &  
        \end{array}\right], \text{ and }
  \left[\begin{array}{rrrr}
            & 0  & 0  & 1 \\
          0 &    &    &   \\
         -1 &    &    &  
        \end{array}\right];
\]
whereas $S$ is spanned by
\[
  \left[\begin{array}{rrrr}
            &    &    &  \\
          1 &    &    &  \\
            & -1 & 0  & 0 
        \end{array}\right],\ 
  \left[\begin{array}{rrrr}
          1 &   & &  \\
            & \omega^2 & &  \\
            &   & \omega & 0 
        \end{array}\right],\ 
  \left[\begin{array}{rrrr}
          0 & 1 &   &  \\
            &   & 1 &  \\
            &   &   & 1
        \end{array}\right],\ 
  \left[\begin{array}{rrrr}
          0 & 1 & &  \\
            &   & \omega &  \\
            &   &        & \omega^2
        \end{array}\right],\ 
  \left[\begin{array}{rrrr}
          0 & 0 & 1  &   \\
            &   &    & 1 \\
            &   &    &  
        \end{array}\right], \text{ and }
  \left[\begin{array}{rrrr}
            & 0  & 0  & 1 \\
          0 &    &    &   \\
          1 &    &    &  
        \end{array}\right].
\]

Since these twelve vectors are clearly orthogonal, the subspaces $R$ and $S$
are each of dimension $6$, and orthogonal to each other; the proof that they don't
contain a product state is as follows: the first five vectors of $R$ and $S$
span respective $5$-dimensional subspaces $R_0$ and $S_0$. Notice that
they are entirely symmetric to each other, and that they don't contain product
states by the arguments of~\cite{sub-schmidt}. Also, the sixth vector is clearly
not product in either case. Hence, to obtain a product vector in $R$, say (the 
argument for $S$ is very similar),
we need to form the sum of the sixth vector with an element from $R_0$:
\[\begin{split}
  M &= \a
       \left[\begin{array}{rrrr}
              &   &    &  \\
            1 &   &    &  \\
              & 1 & 0  & 0 
       \end{array}\right]
      +\b
       \left[\begin{array}{rrrr}
            1 &   & &  \\
              & 1 & &  \\
              &   & 1 & 0 
       \end{array}\right]
      +\g
       \left[\begin{array}{rrrr}
          1 &        & &  \\
            & \omega & &  \\
            &        & \omega^2 & 0 
       \end{array}\right]
      +\d
       \left[\begin{array}{rrrr}
          0 & 1 & &  \\
            &   & \omega^2 &  \\
            &   &          & \omega
       \end{array}\right]
      +\ve
       \left[\begin{array}{rrrr}
          0 & 0 & 1  &    \\
            &   &    & -1 \\
            &   &    &  
       \end{array}\right]                \\
     &\phantom{===========================================}  
      +\left[\begin{array}{rrrr}
            & 0  & 0  & 1 \\
          0 &    &    &   \\
         -1 &    &    &  
       \end{array}\right]                \\
   &= \left[\begin{array}{rrrr}
         \b+\g & \d          & \ve           &    1 \\
         \a    & \b+\omega\g & \omega^2\d    & -\ve \\
         -1    & \a          & \b+\omega^2\g & \omega\d
       \end{array}\right]
\end{split}\]
For this to be a product vector, all its $2\times 2$-minors have to vanish,
but we need to look at only a few to obtain a contradiction:
the minors $\{1,2\}\times\{3,4\}$, $\{1,3\}\times\{2,4\}$ and
$\{2,3\}\times\{1,4\}$ imply
\(
  0 = -\ve^2 - \omega^2\d
    = \omega\d^2 - \a
    = \omega\a\d - \ve
\),
which in turn allow us to express all other variables in terms of $\ve$:
\[
  \d = -\omega\ve^2,\quad
  \a =  \omega\d^2 = \ve^4,\quad 
  \ve = \omega\a\d = -\omega^2 \ve^6,
\]
leaving for $\ve$ only the possibilities of being $0$ or a fifth root
of $-\omega^2$. If $\ve=0$, so are $\a$ and $\d$, and in this case
the $\{1,3\}\times\{1,3\}$-minor is non-vanishing. Hence we continue with
$\ve^5 = -\omega^2$, and look at the minors $\{1,3\}\times\{1,4\}$, 
$\{1,2\}\times\{2,4\}$ and $\{1,3\}\times\{3,4\}$:
these yield the constraints
\[
  0 = (\b+\g)\omega\d + 1
    = -\d\ve - (\b+\omega\g)
    = \omega\d\ve - (\b+\omega^2\g),
\]
in other words
\[
  \b+\g         = -\omega^2/\d = \omega/\ve^2 = -\omega^2\ve^3,\quad
  \b+\omega\g   = -\d\ve = \omega\ve^3,\quad
  \b+\omega^2\g = \omega\d\ve = -\omega^2\ve^3,
\]
which implies $\g=0$ and $\b = -\omega^2\ve^3$ from the 1st and 3rd
equation, but then the 2nd contradicts by demanding $\b = \omega\ve^3$.

Hence, in conclusion, $R$ cannot contain a product state, and the
argument for $S$ is similar in nature.
\qed

\section{Larger R\'{e}nyi parameter}
\label{sec:larger-p}
Now we can use the explicit pair of channels constructed in the
previous section to look for larger values of $p$ for which 
additivity of $S_p^{\min}$ is violated.
The simplest thing is to take the Choi-Jamio\l{}kowski states
to be the normalised projections onto the subspaces $R$ and $S$,
respectively. However, we may clearly take \emph{any} state of
rank $6$ supported on the respective subspace to obtain a bona fide
generalised Choi-Jamio\l{}kowski state.
We performed some numerics in both cases: for the first (Choi-Jamio\l{}kowski
states proportional to the subspace projections),
using $S\bigl(({\cal N}\ox{\cal N}')\Phi_3\bigr)$ 
as an upper bound of $S_p^{\min}({\cal N}\ox{\cal N}')$
and numerical calculations of $S_p^{\min}({\cal N})$ and
$S_p^{\min}({\cal N}')$, we see violations of additivity for values
of $p$ up to $\approx 0.096$.

\begin{figure}[ht]
  \includegraphics[width=8cm,viewport=50 50 600 450,clip]{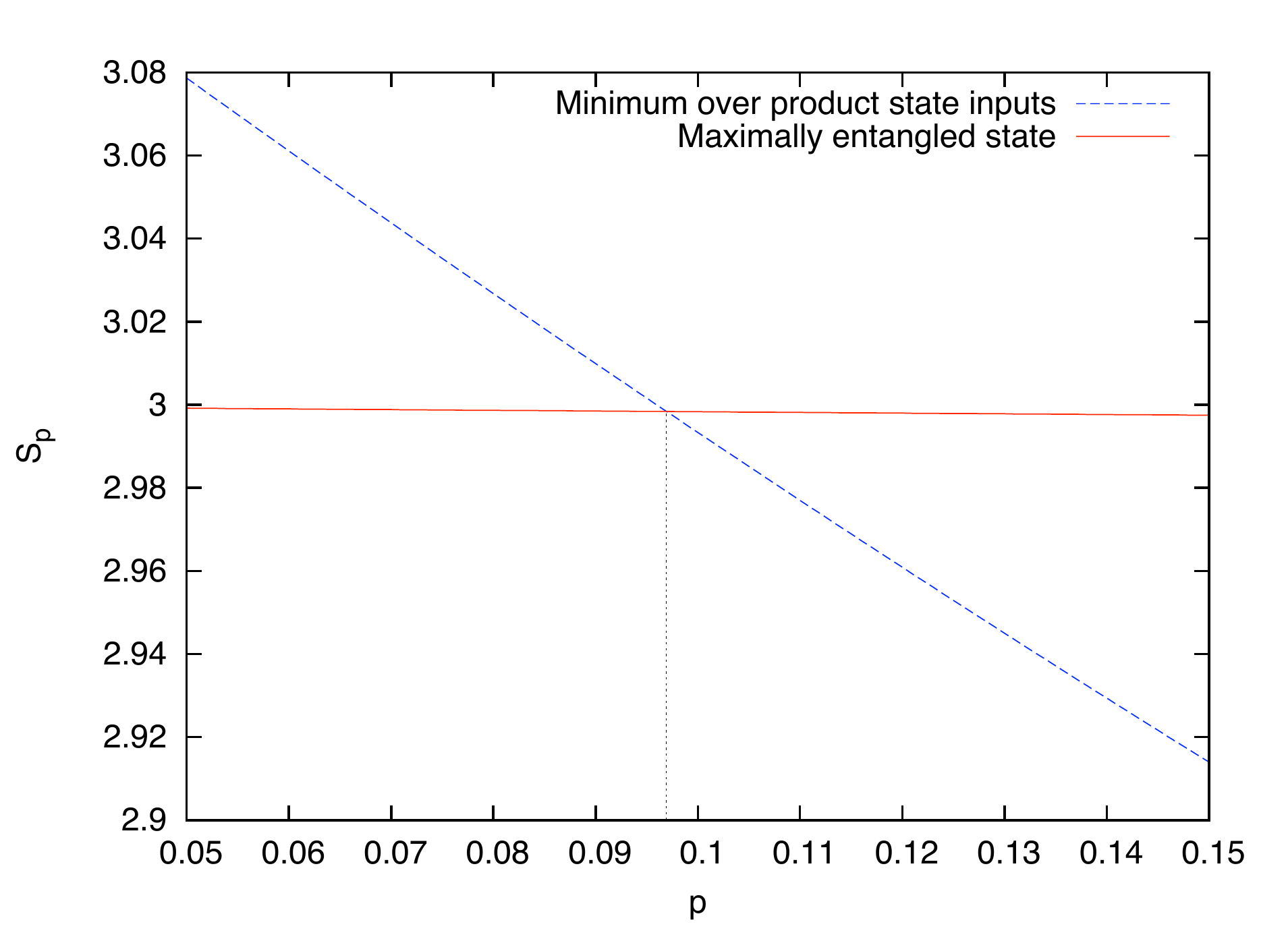}
  \includegraphics[width=8cm,viewport=50 50 600 450,clip]{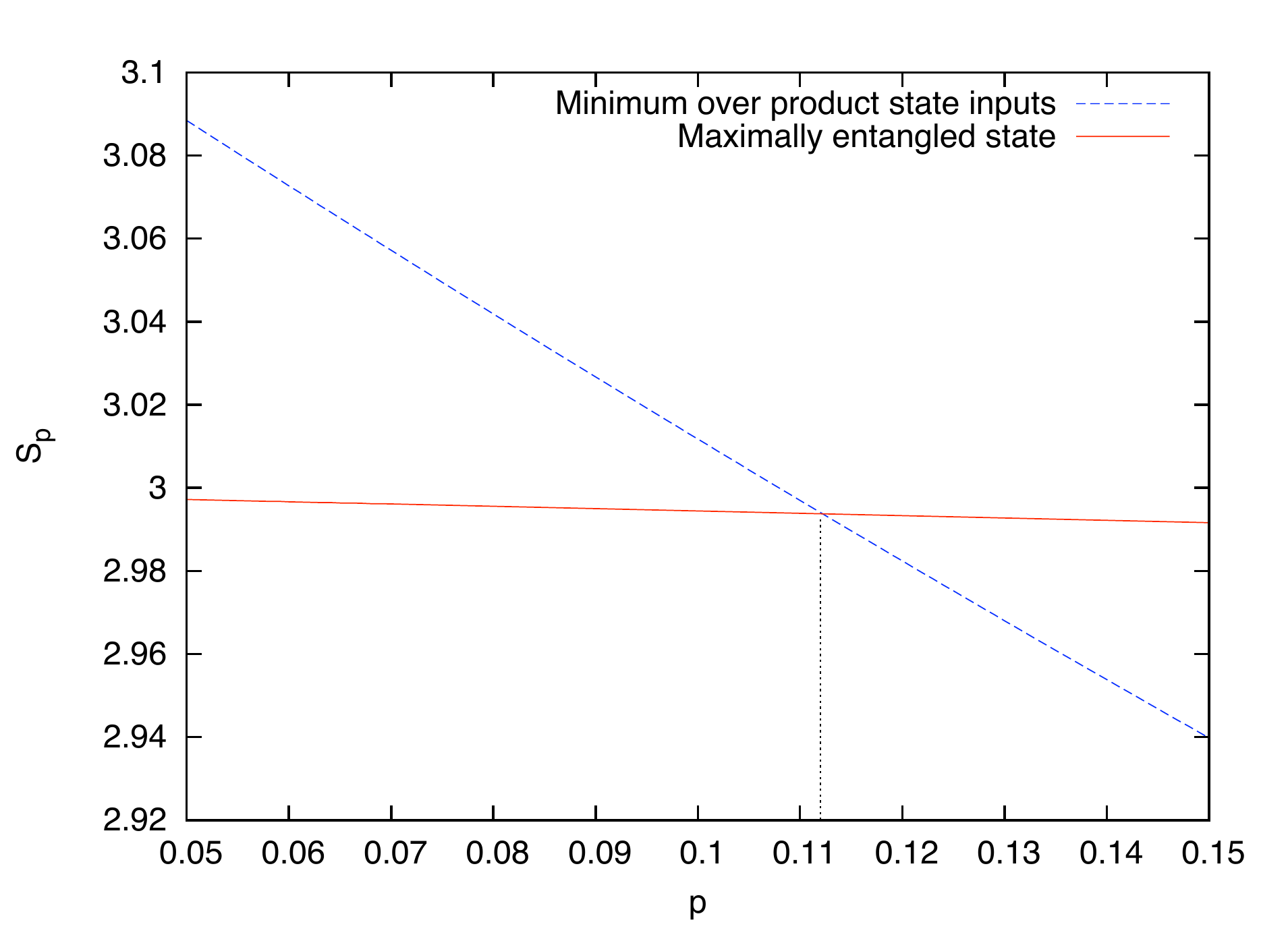}
  \caption{Plots of the output entropy of the tensor product channel with the
    input state corresponding to the maximally entangled state (red line, shallow slope),
    versus the numerically obtained minimum when
    restricted to tensor product input states (blue line, steep slope).
    On the left the Choi-Jamio\l{}kowski states are simply the normalised
    projections of the subspaces $R$ and $S$; on the right, one choses
    appropriately weighted density operators with support $R$ and $S$,
    respectively.}
\end{figure}

For the second, it turns out that a very good choice is 
to have $\rho_{AB}$ and $\sigma_{A'B'}$ to be diagonal in
the above bases of $R$ and $S$, respectively, with specific
probability weights obtained by another numerical search.
The weights of the basis vectors of $R$ and $S$, in the above order,
are
\[\begin{array}{llllll}
    0.172776, & 0.118738, & 0.199229, & 0.136705, & 0.306899, & 0.0656529,\ \text{ and} \\
    0.344911, & 0.124908, & 0.120721, & 0.156968, & 0.162754, & 0.089738,
\end{array}\]
respectively.
This results in a numerical violation of additivity for $p$ up to
$\approx 0.112$. We see no reason to believe that this value
should be the limit of additivity violations.

\medskip
To obtain a rigorous interval $[0;p_0]$ of violations of additivity,
we turn to the ideas of measure concentration explored in~\cite{aspects}
in the context of quantum information theory. We will not make everything
explicit, but the idea is as follows: we need to put rather tight lower
bounds on $S_p^{\min}$ of the two individual channels; in fact, each of the two
channels ${\cal N}$, ${\cal N}'$ is individually random from the class
of channels with Stinespring dilations $A \hookrightarrow B \ox \CC^{d_Ad_B/2}$
(random meaning: according to the unitary invariant measure on
$B \ox \CC^{d_Ad_B/2}$). For the output properties of each of the channels,
only the embedded $d_A$-dimensional subspace $S,\,S' < B \ox \CC^{d_Ad_B/2}$ is
relevant, which is a random subspace in the same sense~\cite{aspects}.

Now in~\cite{aspects}, Lemmas III.4 and III.6, it is shown that the spectrum
of all states in a random subspace $S < B \ox \CC^{d_Ad_B/2}$ is
tightly concentrated around the value $1/d_B$, for large enough
dimensions $d_A$ and $d_B$ such that $d_A \gg d_B \geq \Omega(\log d_A)$.
I.e., with high probability
the minimum Schmidt coefficient of any state in $S$, $S'$ is, say,
$\geq \frac{1}{2}\frac{1}{d_B}$. In other words, the output states of the
channels have spectrum bounded away from $0$ by this amount.
Then for $0\leq p <1$, clearly,
\[\begin{split}
  S_p^{\min}({\cal N}),\ S_p^{\min}({\cal N}')
             &> \frac{1}{1-p} \log \left( d_B \left(\frac{1}{2d_B}\right)^p \right) \\
             &= \log d_B + \frac{p}{1-p}\log\frac{1}{2}
              = \log d_B - \frac{p}{1-p}.
\end{split}\]
However,
\[
  S_p^{\min}\bigl( {\cal N}\ox{\cal N}' \bigr)
      \leq S_0^{\min}\bigl( {\cal N}\ox{\cal N}' \bigr)
      \leq \log\bigl( d_B^2 - 1 \bigr)
      =    2\log d_B + \log\left( 1-\frac{1}{d_B^2} \right).
\]
In conclusion, a violation is obtained as soon as
\[
  \frac{2p}{1-p} \leq -\log\left( 1-\frac{1}{d_B^2} \right),
\]
which follows if $p \leq \frac{1}{1+ 2\ln 2 d_B^2}$. We omit here any estimate
of the $d_B$ required in the above concentration of reduced state spectrum,
which depends on the exact constants one uses in the probability bounds,
but it is possible to get $p_0$ in the range of $10^{-3}$ to $10^{-2}$ 
by this approach.

There is yet another way to get rigorous estimates of $p_0$ for every
example, like for the explicit construction in the previous section.
Namely, get a lower bound on the minimum minimal eigenvalue of an output state of the
single copy channel ${\cal N}$, which can be relaxed to a convex optimisation
problem, and then use the argument above.

In detail, consider the usual Choi-Jamio\l{}kowski operator of the channel,
$\Omega_{AB} = (\id\ox{\cal N})\Gamma$, with $\ket{\Gamma} = \sum_{i=1}^{d_A} \ket{i}\ket{i}$.
Then,
${\cal N}(\varphi) = \tr_A \bigl[ \Omega_{AB} (\overline{\varphi}\ox\1) \bigr]$
(see Appendix A), and
\[\begin{split}
  \min_{\varphi} \lambda_{\min}\bigl( {\cal N}(\varphi) \bigr)
           =    \min_{\varphi,\psi} \tr\bigl[ \Omega_{AB} (\varphi\ox\psi) \bigr] 
          &=    \min_{\rho\text{ separable}} \tr\bigl[ \Omega_{AB} \rho \bigr]    \\
          &\geq \min_{\rho\text{ PPT}} \tr\bigl[ \Omega_{AB} \rho \bigr].
\end{split}\]
The latter is a semidefinite program, so duality theory will yield rigorous
lower bounds on the minimum minimal eigenvalue of an output state.
Doing that for our example in Section~\ref{sec:extension}, yields again
a rather poor bound for $p_0$ of the order $10^{-2}$.

\section{Discussion}
\label{sec:discussion}
After the disproof of the additivity conjecture for $S_p^{\min}$
at $p>1$, and the close shave by which the original and main conjecture
at $p=1$ has escaped, some hope was raised that one could prove
additivity for $p<1$, and hence by taking the limit for $p=1$.
This suggestion didn't seem so unreasonable after King~\cite{king4}
showed additivity if one of the channels is entanglement-breaking.
Also, it can be seen quite easily that arbitrary
numbers of copies of the Holevo-Werner channel~\cite{HW} obey
additivity for $p \leq 1$, via the result of~\cite{AlickiFannes}.
In this respect the log of the minimum
output rank, $S_0^{\min}$, took prominence as an important test
case, and the finding of a counterexample here is putting into doubt
possible programmes to prove the ``standard additivity conjectures''
by approaching $p=1$ from below.

We feel that, with the minimum output rank not multiplicative, it is
rather unlikely that any of the $S_p^{\min}$ for $p<1$ should be additive.
It is to be noted however, that the present technique doesn't really
yield massive violations of additivity, even at $p=0$, and presumably
less so at other $0<p<1$. This is in contrast to what one observes
at $p>1$~\cite{winter-add,hayden-add}, but it can be understood pretty
well in terms of control by the random selection to engineer a certain
conspiracy between the two channels: while for $p>1$ we only
need to fix one large eigenvalue of the two-copy output corresponding to the
maximally entangled input state, at $p<1$
(and most extremely so at $p=0$) \emph{all} non-zero eigenvalues
are relevant, and even to make $d$ of them zero exhausts the
possibilities of the random selection performing well on the
single-copy level. It is amusing to note, however, that we still exploit
the peculiar symmetries, and indeed the multiplicativity, of the maximally
entangled state to construct a violation.

It is our hope that the present work will spark the search for further 
counterexamples, potentially finding a unified principle behind the
constructions for $p>1$ and $p<1$ -- and eventually helping to decide
the original additivity conjecture(s) at $p=1$. Note that the construction
presented here and in~\cite{winter-add,hayden-add} share already a couple
of important traits.
First, the candidate channels are individually random
from the unitary invariant ensemble of Stinespring dilations with fixed
input, output and environment dimensions -- to get strong lower bounds
on the minimum output entropy. Second, the pair of channels
is chosen to be in some fixed relation to each other, so as to make
the output state corresponding to the maximally entangled input (or, in our
case, something very close to it) special; for $p>1$ we want it to have an
unusually large eigenvalue (which is why we choose the channels to be
complex conjugate to each other), here we want an eigenvalue to vanish
(which is why we impose orthogonality on the Choi-Jamio\l{}kowski states).
The possible extension or unification of the constructions thus is
not so much how they are individually selected, but has
to address the way the two channels are related to each other.

\bigskip
{\bf Note added.} After this work was presented at the AQIS'07 
workshop in Kyoto (September 2007), Duan and Shi~\cite{DuanShi} used
the methodology of our explicit construction in their surprising
results on quantum zero-error capacity; they also exhibit a single
channel from $4$ to $4$ dimensions violating additivity of $S_0^{\min}$
-- as opposed to our using the tensor product of two different channels --
in the sense that
$S_0^{\min}\bigl( {\cal N}^{\ox 2} \bigr) < 2 S_0^{\min}({\cal N})$.

\acknowledgments
The authors thank Patrick Hayden, Richard Low, Koenraad Audenaert,
Runyao Duan and Yaoyun Shi
for their interest in the present work, and encouraging as
well as interesting discussions.

AWH and AW acknowledge support by the European Commission under a 
Marie Curie Fellowship (ASTQIT, FP-022194).
TC, AWH, AM and AW acknowledge support through the
integrated EC project ``QAP'' (contract no.~IST-2005-15848),
as well as by the U.K.~EPSRC, project ``QIP IRC''.
AH was furthermore supported by the Army Research Office under grant W9111NF-05-1-0294.
AM was additionally supported by the EC-FP6-STREP network QICS.
DL thanks NSERC, CRC, CFI, ORF, MITACS, ARO, and CIFAR for support.
AW furthermore acknowledges support through an Advanced Research
Fellowship of the U.K.~EPSRC and a Royal Society Wolfson Research
Merit Award.


\appendix

\section{Choi-Jamiolkowski states}
Here we give a detailed explanation of eq. (\ref{eq:c-j}) by
describing how the channel can be recovered from our non-standard
Choi-Jamio\l{}kowski operator.

Recall how to reconstruct the channel from the ``standard''
Choi-Jamio\l{}kowski operator
$\Omega_{AB} = (\id\ox{\cal N})\Gamma$, where
$\ket{\Gamma} = \sum_{i=1}^{d_A} \ket{i}\ket{i}$. The key is the identity
\[
  {\cal N}(\varphi) = \tr_A \bigl[ \Omega_{AB} \bigl( \overline{\varphi} \ox \1 \bigr)  \bigr],
\]
with the complex conjugation with respect to the basis $\{\ket{i}\}_{i=1}^{d_A}$
denoted by $\overline{\cdot}$.

Now, if we have any entangled state $\ket{\Psi}$ of maximal Schmidt rank, it
has a Schmidt form
\[
  \ket{\Psi} = \sum_{i=1}^{d_A} \sqrt{\lambda_i} \ket{e_i}_A \ket{f_i}_{A'},
\]
with local bases $\{\ket{e_i}_A\}_{i=1}^{d_A}$ and $\{\ket{f_i}_{A'}\}_{i=1}^{d_A}$,
and strictly positive Schmidt coefficients $\lambda_i > 0$.
This means that $\Psi_A = \tr_{A'} \Psi_{AA'} = \sum_{i=1}^{d_A}
 \lambda_i \proj{e_i}$ has
full rank (in particular it is invertible), so its inverse is well-defined, and
\[
  \left( \Psi_A^{-1/2}\ox \1 \right)\ket{\Psi}_{AA'} = \sum_{i=1}^{d_A}
  \ket{e_i}\ket{f_i}. 
\]
Thus, introducing the unitary basis change
$U: \ket{e_i} \mapsto \ket{\overline{f_i}}$, we finally get
\[
  \left( U\Psi_A^{-1/2}\ox \1 \right)\ket{\Psi}_{AA'} = \sum_{i=1}^{d_A}
 \ket{\overline{f_i}}\ket{f_i}
  = \sum_{i=1}^{d_A} \ket{i}\ket{i} = \ket{\Gamma},
\]
due to the $V\ox\overline{V}$-invariance of $\ket{\Gamma}$ for unitary $V$.

So, since the mapping from $\Psi$ to $\Gamma$ only acts on $A$
while the Choi-Jamio\l{}kowski mapping acts only on $B$, and using the
fact that $\rho_A = \Psi_A$ for the generalised Choi-Jamio\l{}kowski
state $\rho_{AB} = (\id\ox{\cal N})\Psi_{AA'}$, we finally find
that we can recover the ``standard'' operator $\Omega_{AB}$ as
\[
  \Omega_{AB} = \left( U\Psi_A^{-1/2}\ox \1 \right) \rho_{AB} \left( \Psi_A^{-1/2}U^\dagger\ox \1 \right).
\]
In other words, using the above identities, the channel can be written
\be\begin{split}
  \label{eq:gen-CJ-mapping}
  {\cal N}(\varphi) &= \tr_A \bigl[ \Omega_{AB} \bigl( \overline{\varphi} \ox \1 \bigr)  \bigr] \\
                    &= \tr_A \left[ \left( U\Psi_A^{-1/2}\ox \1 \right) 
                                              \rho_{AB}
                                    \left( \Psi_A^{-1/2}U^\dagger\ox \1 \right)
                                              \bigl( \overline{\varphi} \ox \1 \bigr) \right] \\
                    &= \tr_A \left[ \rho_{AB}
               \left( \rho_A^{-1/2}U^\dagger\overline{\varphi}U\rho_A^{-1/2} \ox \1 \right) \right],
\end{split}\ee
which is eq.~(\ref{eq:c-j}) needed in the proof of Theorem~\ref{thm:main}.
\qed

Different $U$ correspond to choosing different initial reference states $\Psi$
with the same Schmidt spectrum,
with respect to which to formulate the Choi-Jamio\l{}kowski isomorphism. 
Since in our random selection argument we don't mention $\Psi$ to begin with, we
are free to put the unitary to $U=\1$.


\begin{thebibliography}{99}
  \bibitem{shor} P. W. Shor, ``Additivity of the classical capacity of entanglement-breaking
    quantum channels'', J. Math. Phys. {\bf 43}:4334-4340 (2002); arXiv:quant-ph/0201149. 

  \bibitem{king1} C. King, ``Maximization of capacity and p-norms for some product channels'',
    arXiv:quant-ph/0103086 (2001). 

  \bibitem{king2} C. King, ``Additivity for unital qubit channels'',
    J. Math. Phys. {\bf 43}:4641-4653 (2002); arXiv:quant-ph/0103156. 

  \bibitem{king3} C. King, ``The capacity of the quantum depolarizing channel'',
    IEEE Trans. Inf. Theory {\bf 49}:221-229 (2003); arXiv:quant-ph/0204172. 

  \bibitem{king4} C. King, announced at the 1st joint AMS-PTM meeting,
    Warsaw 31 July -- 3 Aug 2007.

  \bibitem{HW} A. S. Holevo, R. F. Werner, ``Counterexample to an
    additivity conjecture for output purity of quantum channels'',
    J. Math. Phys. {\bf 43}:4353Ð4357 (2002); arXiv:quant-ph/0203003.

  \bibitem{winter-add} A. Winter, ``The maximum output p-norm of
    quantum channels is not multiplicative for any p$>$2'',
    arXiv:0707.0402[quant-ph] (2007).

  \bibitem{hayden-add} P. Hayden, ``The maximal p-norm
    multiplicativity conjecture is false'', arXiv:0707.3291[quant-ph]
    (2007).

  \bibitem{sub-schmidt} T. Cubitt, A. Montanaro, A. Winter, ``On the
    dimension of subspaces with bounded Schmidt rank'', to appear
    in J. Math. Phys; arXiv:0706.0705[quant-ph] (2007).

  \bibitem{Eisenbud} D. Eisenbud ``Linear Sections of Determinantal Varieties'',
    Amer. J. Math. {\bf 110}(3):541-575 (1988).
  
  \bibitem{Landsberg} B. Ilic, J. M. Landsberg, ``On symmetric degeneracy loci, 
    spaces of symmetric matrices of constant rank and dual varieties'',
    Math. Ann. {\bf 314}:159Ð174 (1999).
  
  \bibitem{WalgateScott} J. Walgate, A. J. Scott, ``Generic local
    distinguishability and completely entangled subspaces'',
    arXiv:0709.4238[quant-ph] (2007).

  \bibitem{aspects} P. Hayden, D. Leung, A. Winter, ``Aspects of generic entanglement'',
    Comm. Math. Phys. {\bf 265}:95-117 (2006); arXiv:quant-ph/0407049.

  \bibitem{AlickiFannes} R. Alicki, M. Fannes, ``Note on Multiple
    Additivity of Minimal R\'{e}nyi Entropy Output of the
    Werner-Holevo Channels'', Open Systems Inf. Dyn {\bf
      11}(4):339-342 (2004); arXiv:quant-ph/0407033.

  \bibitem{DuanShi} R. Y. Duan, Y. Shi, ``Entanglement between Two Uses of a Noisy 
    Multipartite Quantum Channel Enables Perfect Transmission of Classical Information'',
    arXiv:0712.3700[quant-ph] (2007).

\end{thebibliography}
\end{document}